\documentclass[11pt]{article}
\usepackage[utf8]{inputenc}
\usepackage{graphicx}
\usepackage{fullpage}
\usepackage{amsmath,amsthm,amssymb}
\usepackage{color}
\usepackage{verbatim}
\usepackage{mathabx}
\usepackage{algorithm}
\usepackage{algpseudocode}
\usepackage{xspace}

\newtheorem{theorem}{Theorem}[section]
\newtheorem{lemma}[theorem]{Lemma}

\newcommand{\cross}{\times}

\newcommand{\UN}{\mathsf{UN}}
\newcommand{\IN}{\mathsf{in}}
\newcommand{\OUT}{\mathsf{out}}

\newcommand{\tw}{\mathsf{tw}}
\newcommand{\dtw}{\mathsf{dtw}}

\newcommand{\ndp}{\text{NDP}\xspace}
\newcommand{\edp}{\text{EDP}\xspace}
\newcommand{\anf}{\text{ANF}\xspace}
\newcommand{\sdndp}{\text{Sym-Dir-NDP}\xspace}

\newcommand{\sdanf}{\text{Sym-Dir-ANF}\xspace}

\title{Routing Symmetric Demands in  Directed Minor-Free Graphs\\
with Constant Congestion}

\author{
Timothy Carpenter\thanks{Dept.~of Computer Science \& Engineering, The Ohio State University, \texttt{carpenter.454@osu.edu}.}
\and
Ario Salmasi\thanks{Dept.~of Computer Science \& Engineering, The Ohio State University, \texttt{salmasi.1@osu.edu}.}
\and
Anastasios Sidiropoulos\thanks{Dept.~of Computer Science, University of Illinois at Chicago, \texttt{sidiropo@uic.edu}.}
}

\date{November 3, 2017}

\begin{document}

\clearpage
\maketitle
\thispagestyle{empty}

\begin{abstract}
The problem of routing in graphs using node-disjoint paths has received a lot of attention and a polylogarithmic approximation algorithm with constant congestion is known for undirected graphs [Chuzhoy and Li 2016] and [Chekuri and Ene 2013].
However, the problem is hard to approximate within polynomial factors on directed graphs, for any constant congestion [Chuzhoy, Kim and Li 2016].

Recently, [Chekuri, Ene and Pilipczuk 2016] have obtained a polylogarithmic approximation with constant congestion on directed planar graphs, for the special case of symmetric demands.
We extend their result by obtaining a polylogarithmic approximation with constant congestion on arbitrary directed minor-free graphs, for the case of symmetric demands.

\end{abstract}

\pagebreak{}
\setcounter{page}{1}

\section{Introduction}

Routing in graphs along disjoint paths is a fundamental problem with numerous applications  in various domains \cite{Aggarwal:1994:ERS:314464.314579, Awerbuch:1994:OAC:1398518.1398993, doi:10.1137/S0097539792232021, Peleg1989, Raghavan:1994:ERA:195058.195119}. Disjoint path problems have been well-studied in both the directed and undirected setting, and it is known that the directed formulations of these problems are generally more difficult to approximate \cite{DBLP:journals/corr/ChuzhoyKN16, Chuzhoy:2007:HRC:1250790.1250816}. The recent work of \cite{Chekuri2015,chekuri2016constant} has brought to light a more tractable formulation of the directed version of these problems by considering routing symmetric demand pairs with constant congestion.

Two of the most well-known and studied disjoint path problems are the node-disjoint paths problem (\ndp) and the edge-disjoint paths problems (\edp). In these problems, the goal is to connect a set of node pairs through node- or edge-disjoint paths in an undirected graph.
It is known that the decision version of \ndp is NP-complete \cite{karp1975},
and it has been shown to be fixed parameter tractable \cite{ROBERTSON199565}.
But there remain gaps in our understanding of their approximability. 
For both \edp and \ndp on $n$-node graphs, the state of the art is an $O(\sqrt{n})$-approximation \cite{Chekuri06ano}, \cite{Kolliopoulos98approximatingdisjoint-path}.
For planar graphs, a slightly better bound of $\tilde{O}(n^{9/19})$-approximation is known \cite{DBLP:journals/corr/ChuzhoyKL16}.
Even for the case of the grid, only a $\tilde{O}(n^{1/4})$-approximation for \ndp is known \cite{chuzhoy_et_al:LIPIcs:2015:5303}.
For hardness of approximation, it is known that both \ndp and \edp are $2^{\Omega\left(\sqrt{\log(n)}\right)}$-hard to approximate, unless all problems in NP have algorithms with running time $n^{\log(n)}$  \cite{DBLP:journals/corr/ChuzhoyKN16}.

Progress has been made on relaxed versions of these problems. One such relaxation is the all-or-nothing flow problem (\anf), where a subset $\mathcal{M}' \subseteq \mathcal{M}$ is routed if there is a feasible multicommodity flow routing one unit of flow for each pair in $\mathcal{M}'$. Poly-logarithmic approximations are known for \anf \cite{Chekuri:2005:MFW:1060590.1060618, Chekuri:2004:AMF:1007352.1007383}.
Another relaxation is to allow some small constant congestion on the nodes or edges. For this relaxation, poly-logarithmic approximations have been obtained for EDP with congestion 2 \cite{DBLP:journals/corr/abs-1208-1272}, and for NDP with congestion $O(1)$ \cite{Chekuri:2013:PAM:2627817.2627841}.

It is natural to extend the study of disjoint path problems to directed graphs. However, these problems are known to be significantly harder on directed graphs.
Even the case of \anf with constant congestion $c$ allowed has an $n^{\Omega(1/c)}$ inapproximability bound \cite{Chuzhoy:2007:HRC:1250790.1250816}.
However, a more tractable case is found by considering symmetric demand pairs.
The study of maximum throughput routing problems in directed graphs with symmetric demand pairs began in \cite{Chekuri2015}. 
In this setting the graph $G$ is directed, and routing a source-destination pair $(s_i, t_i)$ requires finding a path from $s_i$ to $t_i$ and a path from $t_i$ to $s_i$. We let \sdanf be the analogue of \anf, and \sdndp be the analogue of \ndp in this setting. A poly-logarithmic approximation for \sdanf is obtained in \cite{Chekuri2015}.
Subsequently, in \cite{chekuri2016constant} a randomized poly-logarithmic approximation with constant congestion on planar graphs for \sdndp is obtained.

\subsection{Our contribution}
We consider the problem of routing symmetric demands along node-disjoint paths in directed graphs. We refer to this problem as \sdndp. Letting $G = (V, E)$ be a directed graph with unit node capacities and $\mathcal{M} = \left\{(s_1, t_1), \ldots, (s_k, t_k) \right\} \subseteq V \times V$ be a set of source-destination pairs, we say that $(G, \mathcal{M})$ is an instance of \sdndp. Routing a pair $(s_i, t_i)$ requires finding a path from $s_i$ to $t_i$, and from $t_i$ to $s_i$. A solution to an instance of \sdndp is a routing strategy maximizing the number of pairs routed through disjoint paths.
Our contribution generalizes the result from \cite{chekuri2016constant} from the class of directed planar graphs to arbitrary directed minor-free graphs.
We now formally state our results and briefly highlight the methods used. Our main result is the following.

\newtheorem*{thm:free}{Theorem \ref{thm:free}}
\begin{thm:free}
Let $G$ be a $H$-minor free graph. There is a polynomial time randomized  algorithm that, with high probability, achieves an $\Omega\left(\frac{1}{h^7\sqrt{h} \log^{5/2}(n)}\right)$-approximation with congestion $5$ for Sym-Dir-NDP instances in $G$, where $h$ is an integer dependent only on $H$.
\end{thm:free}

The approximation algorithm in this theorem is obtained by extending the algorithm of \cite{chekuri2016constant}.
For an instance $(G, \{(s_1, t_1), \ldots, (s_k, t_k)\})$ of \sdndp, we say the set $\mathcal{T} = \{s_1, \ldots, s_k\} \cup \{t_1, \ldots, t_k\}$ is the set of terminals. Speaking broadly, the algorithm obtained in Theorem \ref{thm:free} consists of the following steps.
\begin{enumerate}
\item Using a multicommodity flow based LP relaxation and the well-linked decomposition of \cite{chekuri2016constant}, reduce to an instance in which the terminals $\mathcal{T}$ are $\alpha$-well-linked for a fixed constant $\alpha$. \label{alg:step1}
\item Find a large routing structure connected to large proportion of the terminals. \label{alg:step2}
\item Use the routing structure to connect a large number of the source-destination pairs. \label{alg:step3}
\end{enumerate}
From here on, we shall refer to the routing structure as the crossbar.
The reduction we use in Step \ref{alg:step1} allows us to reduce an instance of \sdndp to an instance on an Eulerian graph of small maximum degree, and where the terminals are $\alpha$-well-linked. This comes at the cost of then having a randomized algorithm for the original instance. This reduction comes from \cite{chekuri2016constant}, and while there it is used for planar graphs, we were fortunate in that it can also be used for general graphs. The routing scheme of Step \ref{alg:step3} is also thanks to \cite{chekuri2016constant}, and relies on similar crossbar construction. Our main contribution to this line of research is in finding an appropriate crossbar construction for Step \ref{alg:step2}. 

To build our crossbar, we would ideally find a ``flat'' grid minor so that some constant fraction of the terminal pairs can be routed along node-disjoint paths to the interface of the grid minor (a ``flat'' grid minor is one in which the grid minor is connected with the rest of the graph only through the outer face). Then we would have the following sets of node-disjoint paths along which to route the terminal pairs: the paths from the terminals to the interface, the paths from terminals to terminals implied by the node-well-linked property of the terminals, the concentric cycles of the grid minor, and the paths connecting the outermost and innermost cycles of the grid minor. From these, just as in \cite{chekuri2016constant} we can construct a routing scheme with congestion 5. 
To find a suitable flat grid minor, we combine results of \cite{chekuri2013approximation} and \cite{thomassen1997simpler} to show that flat grid minors of a suitable size can be found. We then show that if for the flat grid minor produced we cannot route a large enough fraction of the terminals to the interface then there exists some vertex which can be eliminated from the graph without destroying a potential solution to the problem. Thus, we find and test flat grid minors until one suitable to be used in the crossbar is found.

\subsection{Organization}
In Section \ref{sec:notation}, we introduce the definitions and notion we will use throughout the paper.
After that, in Section \ref{sec:polylogarithmic} we present Theorem \ref{thm:main_crossbar}, our main technical result, and Theorem \ref{thm:free}, our main result. Sections \ref{sec:crossbar}, \ref{sec:genus}, \ref{sec:nearly}, and \ref{sec:minorfree} provide the details of the proofs of Theorems \ref{thm:main_crossbar} and \ref{thm:free}.
The proof of Lemma \ref{lem:minor_free_flat} has been deferred to \ref{sec:missing-proofs}, due to its length and technical nature.

\section{Notation and Preliminaries}\label{sec:notation}

We now introduce some notation and definitions that are used throughout the paper.

\paragraph{Directed and undirected graphs.}
From any directed graph $G$ we can obtain an undirected graph $G^\UN$ as follows.
We set $V(G^\UN)=V(G)$ and $E(G^\UN) = \left\{\{u,v\}: (u,v) \in E(G) \lor (v,u) \in E(G) \right\}$.
We refer to $G^\UN$ as the \emph{underlying undirected graph} of $G$.

\paragraph{Flat subgraphs.}
We say that a planar subgraph $H$ of an undirected graph $G$ is \emph{flat} if there exists a planar drawing $\Phi$ of $H$ such that for any $\{u,v\} \in E(H)$, with $u \in V(H)$ and $v \in V(G) \setminus V(H)$, we have that $u$ is on the outer face of $\Phi$.

\paragraph{Well-linked sets.}
Let $G$ be a directed (resp.~undirected) graph. A set $X \subseteq V(G)$ is node-well linked in $G$ if for any two disjoint subsets $Y,Z \subset X$ of equal size, there exist $|Y|$ node-disjoint directed (resp.~undirected) paths from $Y$ to $Z$, such that each vertex in $Y$ is the start of exactly one path, and each vertex in $Z$ is the end of exactly one path. For some $\alpha \in (0,1)$, we say that $X$ is $\alpha$-node well-linked if for any two disjoint subsets $Y,Z \subset X$ of equal size, there exist $|Y|$ directed (resp.~undirected) paths from $Y$ to $Z$ such that no vertex is in more than $1/\alpha$ of these paths; In other words, we allow a node congestion of $1/\alpha$ for these paths.

\paragraph{Directed and undirected treewidth}
For a directed graph $G$, we will denote by $\dtw(G)$ the \emph{directed treewidth of $G$}, and we will denote by $\tw(G^\UN)$ the (undirected) \emph{treewidth of $G^\UN$}.
Directed undirected is a global connectivity measure introduced in \cite{JOHNSON2001138, REED1999222}, and just as undirected treewidth is defined by the minimum size tree decomposition, directed treewidth is defined by the minimum size of what is termed an aboreal decomposition.
Instead of providing the full definitions of directed and undirected treewidth here, we only ask the reader to make a note of the following two important facts:
\begin{itemize}
\item If $G$ is an Eulerian directed graph with max degree $\Delta$, then $\tw(G^\UN) \leq \dtw(G) = O(\Delta \cdot \tw(G^\UN))$ \cite{JOHNSON2001138}.
\item If a directed graph $G$ contains an $\alpha$-well-linked set $X$, then $\dtw(G) = \Omega(\alpha|X|)$ \cite{REED1999222}.
\end{itemize}

\paragraph{Clique-sums.}
Let $G_1$ and $G_2$ be two graphs. A \emph{clique-sum} of $G_1$ and $G_2$ is any graph that is obtained by identifying a clique in $G_1$ with a clique of the same size in $G_2$, and then possibly removing some edges in the resulting shared clique. An \emph{$h$-clique-sum} or \emph{$h$-sum} for short, is a clique-sum where the identified cliques have at most $h$ vertices.

\paragraph{Nearly-embeddable and minor-free graphs.}
We say that a graph is \emph{$(a,g,k,p)$-nearly embeddable} if it is obtained from a graph of Euler genus $g$ by adding $a$ apices and $k$ vortices of pathwidth $p$. We say that a graph is $h$-nearly embeddable if it is $(a,g,k,p)$-nearly embeddable for some $a,g,k,p \leq h$. The following is implicit in \cite{robertson2003graph}.

\begin{theorem}[Robertson and Seymour \cite{robertson2003graph}]\label{thm:robertson}
Let $H$ be any graph. Every $H$-minor-free graph can be obtained by at most $h$-sums of graphs that are $h$-nearly embeddable graphs, where $h$ is a non-negative integer dependent on $H$.
\end{theorem}

\section{The Algorithm for Minor-Free Graphs}\label{sec:polylogarithmic}
We first use the following result of \cite{chekuri2016constant} to reduce the problem to the case of Eulerian graphs with small degrees. Note that this result is stated for planar graphs in \cite{chekuri2016constant}, but the proof does not use planarity, and thus can be stated for general graphs.

\begin{lemma}[Chekuri, Ene \& Pilipczuk \cite{chekuri2016constant}]\label{lem:reduction}
Suppose that there is a polynomial time algorithm for $\Omega(1)$-node-well-linked instances of \sdndp in directed Eulerian graphs of maximum degree $\Delta$ that achieves a $\beta(\Delta)$-approximation with congestion $c$. Then there is a polynomial time randomized algorithm that, with high probability, achieves a $\beta(O(\log^2 k )) \cdot O(\log^6 k)$-approximation with congestion $c$ for arbitrary instances of \sdndp in directed graphs, where $k$ is the number of pairs in the instance.
\end{lemma}

Now we describe how to construct the crossbar in minor-free graphs, assuming that we are given a $m \cross m$ flat grid minor $\Gamma$, for some large enough $m$, and a family of $\lambda m$ node-disjoint paths connecting the set of terminals and the interface of $\Gamma$, for some constant $\lambda$. The following is our main technical result, which is similar to the one in \cite{chekuri2016constant} for planar graphs. The proof is deferred to Section \ref{sec:crossbar}.

\begin{theorem}\label{thm:main_crossbar}
Let $G$ be a directed minor-free graph of maximum in-degree of at most $\Delta$. Let $X$ be an $\alpha$-node-well-linked set in $G$ with $|X| = \Omega\left(\frac{{\Delta}^2}{\alpha}\right)$. Let $m = \Omega\left(\frac{\alpha |X|}{\beta}\right)$, where $\beta$ is dependent on the structure of $G$. Suppose that we can find a $m \cross m$ flat grid minor $\Gamma$ of $G^\UN$, and a family of $\lambda m$ node-disjoint paths connecting $X$ and the interface of $\Gamma$ in $G^\UN$, for some $0 <\lambda \leq 1$.

One can in polynomial time find a set of $\Omega\left(\frac{\alpha |X|} {\beta \Delta}\right)$ concentric directed cycles going in the same direction, sets $Y^+,Y^- \in X$ with $|Y^+| = |Y^-| = \Omega\left(\frac{\alpha^2 |X|} { \beta \Delta^2}\right)$, and families $P^+$ and $P^-$ of node-disjoint paths such that one of the following holds.
\begin{description}
\item{(1)} None of the cycles enclose any vertex of $Y^+ \cup Y^-$, the family $P^+$ consists of $|Y^+|$ node-disjoint paths from $Y^+$ to the innermost cycle, and the family $P^-$ consists of $|Y^-|$ node-disjoint paths from the innermost cycle to $Y^-$.

\item{(2)} All cycles enclose $Y^+ \cup Y^-$, the family $P^+$ consists of $|Y^+|$ node-disjoint paths from $|Y^+|$ to the outermost cycle, and the family $P^-$ consists of $|Y^-|$ node-disjoint paths from the outermost cycle to $Y^-$. 
\end{description}
\end{theorem}

In order to obtain such a crossbar, we need to find a flat grid minor of large enough size. The following Lemma provides us the desired flat grid minor, and the proof is deferred to Section \ref{sec:minorfree}.

\begin{lemma}\label{lem:flat_in_minor_free}
Let $H$ be any graph and let $G$ be an $H$-minor-free directed graph with treewidth $t$. One can in polynomial time find a $r \cross r$ flat grid minor $\Gamma$ in $G^\UN$, with $r = \Omega\left(\frac{t}{h^7\sqrt{h} \log^{5/2}(n)}\right)$, and a family of $r$ node-disjoint paths connecting $X$ and the interface of $\Gamma$, where $h$ is an integer dependent only on the structure of $H$.
\end{lemma}

Once we obtain a crossbar as described above, we can route a large subset of terminal pairs.

\begin{lemma}\label{lem:eulerian}
Given the crossbar described in Theorem \ref{thm:main_crossbar}, one can get an $O\left(\frac{\Delta^2}{\beta \alpha^3}\right)$-approximation algorithm with congestion $5$ for \sdndp in instances for which the input graph is minor-free and Eulerian with maximum in-degree $\Delta$, the set of terminals is $\alpha$-node-well-linked for some $\alpha \leq 1$, and $\beta$ is dependent on the structure of $G$.
\end{lemma}

\begin{proof}
By applying the same routing scheme as in the one in \cite{chekuri2016constant}, we get the desired result.
\end{proof}

Now we are ready to state the main result of this paper.

\begin{theorem}\label{thm:free}
Let $G$ be a $H$-minor free graph. There is a polynomial time randomized  algorithm that, with high probability, achieves an $\Omega\left(\frac{1}{h^7\sqrt{h} \log^{5/2}(n)}\right)$-approximation with congestion $5$ for \sdndp instances in $G$, where $h$ is an integer dependent only on $H$.
\end{theorem}

\begin{proof}
This is immediate by Lemmas \ref{lem:reduction}, \ref{lem:flat_in_minor_free}, \ref{lem:eulerian}, and Theorem \ref{thm:main_crossbar}.
\end{proof}

\section{The Crossbar Construction}\label{sec:crossbar}
In this section we discuss the construction of the crossbar stated in Theorem \ref{thm:main_crossbar}.
Before we give the proof of this Theorem we establish some auxiliary facts. Throughout this subsection, we assume that we are given the input of Theorem \ref{thm:main_crossbar}.

\begin{lemma}\label{lem:undir_cycles}
One can in polynomial time find an integer $r = \Omega\left(\frac{\alpha |X|}{\beta}\right)$ and a sequence of node-disjoint concentric undirected cycles $C_1, C_2, \ldots, C_r$ in $G^{\UN}$, with $C_1$ being the outermost and $C_r$ being the innermost cycle.
\end{lemma}
\begin{proof}
Let $\Gamma$ be a flat $m \times m$ grid minor of $G^\UN$, as given in the input of Theorem \ref{thm:main_crossbar}. By losing a constant factor, we can construct a flat sub-divided $r \cross r$ wall in $G^{\UN}$, with $r = \Omega\left(\frac{\alpha |X|}{\beta}\right)$. Let $t$ be the treewidth of $G^{\UN}$.
Since $X$ is $\alpha$-node-well-linked in $G$, $X$ is also $\alpha$-node-well-linked in $G^{\UN}$. Thus, $t = \Omega\left(\alpha|X|\right)$.
Let $C_1$ be the outermost cycle of $\Gamma$, and for each $i \in \{2, \ldots, r\}$, let $C_i$ be the outermost cycle of $\Gamma \setminus \cup_{1 \leq j < i} V(C_i)$.
\end{proof}

As in \cite{chekuri2016constant}, for a vertex set $Q \subseteq V(G^\UN)$, a vertex $v \notin Q$, and an integer $\ell \geq 2 \Delta$, we say that a vertex set $S$ is a $(v, Q, \ell)$-\emph{isle} if $v \in S$, $G^{\UN}[S]$ is connected, $S \cap Q = \emptyset$, and $|N_{G^{\UN}}(S)| \leq \ell$.

We set isles $S^\OUT$ and $S^\IN$ by choosing an arbitrary vertex $v^\OUT$ in $C_1$, and an arbitrary vertex $v^\IN$ in $C_r$. Letting $\ell = \lfloor r / (4\Delta + 2) \rfloor$, then $S^\OUT$ is the $(v^\OUT, X, \ell)$-isle and $S^\IN$ is the $(v^\IN, X, \ell)$-isle obtained. We also need that $S^\OUT$ and $S^\IN$ are separated by many cycles. For this, we use the following Lemma of \cite{chekuri2016constant}, the proof of which is slightly modified.

\begin{lemma}
The isle $S^\OUT$ does not contain any vertex that is enclosed by $C_{\ell+1}$, and the isle $S^\IN$ does not contain any vertex that is not strictly enclosed by $C_{r-\ell}$.
\end{lemma}
\begin{proof}
The proofs for $S^\IN$ and $S^\OUT$ are symmetrical, so we focus on the case of $S^\OUT$. Assume that $S^\OUT$ contains a vertex enclosed by $C_{\ell+1}$, and we will find a contradiction. Since $v^{out} \in S^\OUT$, $S^\OUT$ is connected in $G^{\UN}$, and $\Gamma$ is a flat wall, it must be that $S^\OUT$ contains a vertex from every cycle $C_i$, $1 \leq i \leq \ell+1$. Since $|N_{G^{\UN}}(S^\OUT)| \leq \ell$, for some $1 \leq i \leq \ell+1$ we have that $V(C_i)$ is completely contained in $S^\OUT$. However, there are $r > \ell$ vertex-disjoint paths in $G^{\UN}$ connecting $C_i$ with $X$. Thus, either $S^\OUT \cap X \neq \emptyset$ or $|N_{G^{\UN}}(S^\OUT)| > \ell$, both of which are contradictions.
\end{proof}

We are almost ready to prove the main result of this section. We will make use of the following Lemma, which is implicit in \cite{chekuri2016constant}. Note that sets ${S'}^\IN$ and ${S'}^\OUT$, the concentric cycles $C'_1, \ldots, C'_{r'}$, and integers $r'$ and $\Delta'$ in the next Lemma are defined for a planar graph $G'$ as described in \cite{chekuri2016constant}. 

\begin{lemma}\label{lem:planar_dir_cycles}
Let $G'$ be an Eulerian, planar directed graph, with sets ${S'}^\IN, {S'}^\OUT$ separated by concentric cycles $C'_1, \ldots, C'_{r'}$, and let $\ell' = \lfloor r' / (4\Delta' + 2) \rfloor$, where $\Delta'$ is the the maximum in-degree of $G'$. Then one can in polynomial time find $\lceil \ell' / 2 \rceil$ node-disjoint directed concentric cycles, all going in the same direction (all clockwise or all counter-clockwise), such that all vertices of ${S'}^\IN$ are strictly enclosed by the innermost cycle, and all vertices of ${S'}^\OUT$ are not enclosed by the outermost cycle, or vice versa, with the roles of ${S'}^\IN$ and ${S'}^\OUT$ swapped.
\end{lemma}

We will use Lemma \ref{lem:planar_dir_cycles} to find concentric cycles in minor-free graphs. We first generalize the notion of \emph{enclosed} for flat grids in non-planar graphs.
Let $H$ be a directed graph with a flat grid minor $\eta$. Let $u^\OUT$ an arbitrary vertex not contained in $\eta$. Let $C$ be some cycle contained within $\eta$. We say a vertex $u$ is \emph{enclosed} by $C$ if all paths in $H^{\UN}$ from $u$ to $u^\OUT$ intersect $C$.
We now find the desired concentric cycles in $G$.

\begin{lemma}\label{lem:main_lem}
One can in polynomial time find $\lceil \ell / 2 \rceil$ node-disjoint directed concentric cycles, all going in the same direction (all clockwise or all counter-clockwise), such that all vertices of $S^\IN$ are enclosed by the innermost cycle, and all vertices of $S^\OUT$ are not enclosed by the outermost cycle, or vice versa, with the roles of $S^\IN$ and $S^\OUT$ swapped.
\end{lemma}
\begin{proof}
We proceed by creating $G'$ from $G$ as follows. Let
\[
Z = \left\{v \in V(G) : \mbox{$v \in V(C_1)$ or $v$ is not in the component of $G \setminus V(C_1)$ containing $C_2$}\right\}.
\]
Let 
\[
G' = G / Z,
\]
i.e. $G'$ is the graph created by identifying all vertices in $Z$ to a single vertex $z$. Since $G$ is Eulerian, $G'$ is also Eulerian. Furthermore, we can delete any self-loops on $z$, and $G'$ is still Eulerian. Since $C_1, \ldots, C_r$ are contained within a flat grid minor of $G$, $G'$ is therefore a planar graph. The only impediment to directly applying Lemma \ref{lem:planar_dir_cycles} is that the in-degree $\delta$ of $z$ might be greater than $\Delta$. We can eliminate this by replacing $z$ with a path $P$ of length $\delta$, with edges directed both ways between adjacent vertices. We then connect the vertices formerly connected to $z$ to vertices in $P$, maintaining the planarity of $G'$. Then, to restore $G'$ as an Eulerian graph, for the vertices in $P$ with an imbalance between in- and out-degree we can create a new edge (See Figure \ref{fig:Eulerian}).

\begin{figure}[H]
\begin{center}
\scalebox{0.45}{\includegraphics{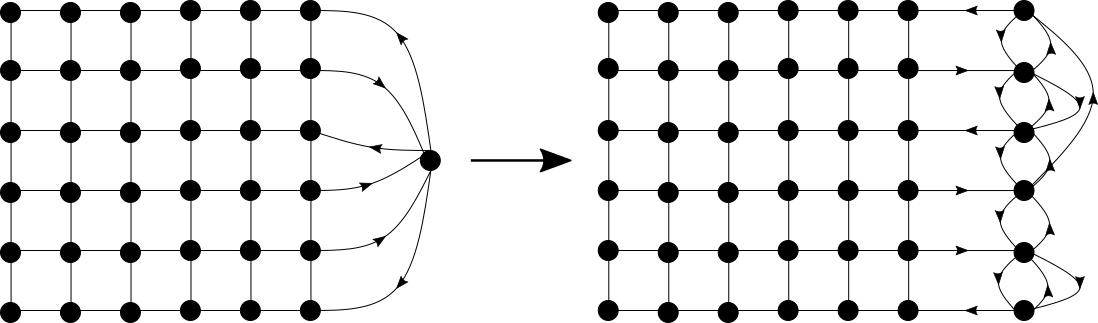}}
\caption{Maintaining a Eulerian graph with bounded degree.}
\label{fig:Eulerian}
\end{center}
\end{figure}

After these modifications, $G'$ is an Eulerian, planar digraph with maximum in-degree $\Delta$. Let $S'_\IN = S^\IN$ and $S'_\OUT = \left( S^\OUT \cap V(G') \right) \cup \{z\}$. We now apply Lemma \ref{lem:planar_dir_cycles} using $G'$, $S'_\IN$, and $S'_\OUT$ to find $\lceil \ell / 2 \rceil$ node-disjoint directed concentric cycles, all going in the same direction, and all vertices of $S'_\IN$ are strictly enclosed by the innermost cycle, and all vertices of $S'_\OUT$ are not enclosed by the outermost cycle. Clearly, each of these cycles exists in $G$, all vertices of $S^\IN$ are strictly enclosed by the innermost cycle, and all vertices of $S^\OUT$ are not enclosed by the outermost cycle.  
\end{proof}

We are now ready to obtain the proof of Theorem \ref{thm:main_crossbar}.

\begin{proof}[Proof of Theorem \ref{thm:main_crossbar}]
By Lemma \ref{lem:main_lem}, we can finish the construction of the crossbar with the same argument as in \cite{chekuri2016constant}.
\end{proof}

\section{Graphs of Bounded Genus}\label{sec:genus}
In this section we describe an algorithm to construct a flat grid minor $\Gamma$ of large enough size in graphs of bounded genus. The following is implicit in the work of Chekuri and Sidiropoulos \cite{chekuri2013approximation}.

\begin{lemma}\label{lem:grid_minor}
Let $G$ be an undirected graph of Euler genus $g \geq 1$, with treewidth $t \geq 1$. There is a polynomial time algorithm that computes a $r' \cross r'$-grid as a minor, with $r' = \Omega\left(\frac{t}{g^3\log^{5/2}n}\right)$. Furthermore, the algorithm does not require a drawing of $G$ as part of the input.
\end{lemma}

We need to find a flat grid a minor for our purpose. Thomassen in \cite{thomassen1997simpler} shows that if a graph of genus $g$ contains a $m \times m$-grid as a minor, then it contains a $k \times k$ flat grid minor, where $m > 100k\sqrt{g}$. With some minor modifications, we can use this result to obtain the following.

\begin{lemma}\label{lem:thomassen}
Let $G$ be an undirected graph of Euler genus $g \geq 1$, and let $H$ be a $m \times m$ grid minor of $G$. Let $k < \frac{m}{100\sqrt{g}}$ be an integer. Then one can compute a $k \times k$ flat grid minor of $G$ in polynomial time.
\end{lemma}

\begin{proof}
We follow a similar approach as in \cite{thomassen1997simpler}. We need to construct a family of pairwise disjoint subgraphs $Q_1,Q_2,\cdots,Q_{2g+2}$ of $H$, satisfying the following conditions.
\begin{description}
\item{1.}
Each $Q_i$ is a $k\times k$ sub-grid of $H$.
\item{2.}
For any $i,j$ with $1 \leq i < j \leq 2g+2$, we have the following. If $x_i$ and $x_j$ are on the outer cycles of $Q_i$ and $Q_j$ respectively, and they have neighbors $y_i \in V(H) \setminus V(Q_i)$ and $y_j \in V(H) \setminus V(Q_j)$ respectively, then $H$ has a path $P_{ij}$ from $x_i$ to $x_j$ such that $V\left(P_{ij}\right) \cap \left(\bigcup\limits_{r = 1}^{2g+2} V(Q_r)\right) = \{x_i,x_j\}$.
\end{description}
Since we have $m > 100k\sqrt{g}$, this can be easily done as shown in Figure \ref{fig:grid}, which completes the proof.
\end{proof}

\begin{figure}
\begin{center}
\scalebox{0.25}{\includegraphics{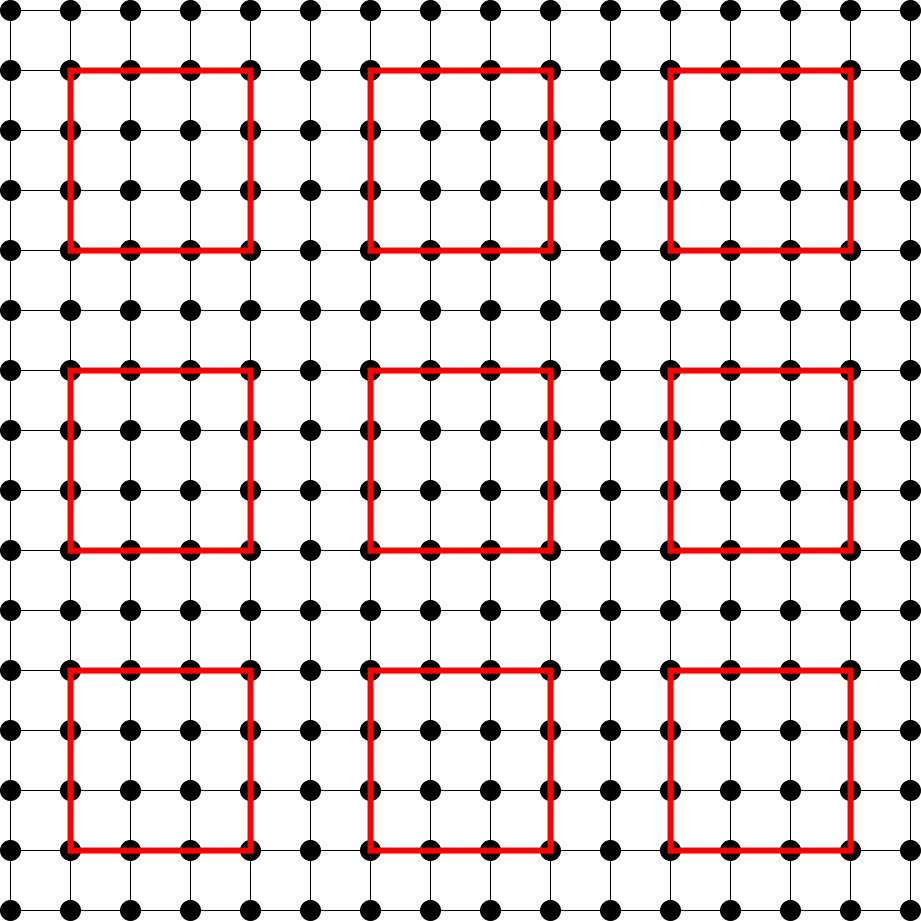}}
\caption{Decomposition of the grid minor}\label{fig:grid}
\end{center}
\end{figure}

\begin{lemma}\label{lem:flat_grid_minor}
Let $G$ be an undirected graph of Euler genus $g \geq 1$, with treewidth $t \geq 1$. There exists a polynomial time algorithm that computes a $r \times r$-grid as a minor, with $r = \Omega\left(\frac{t}{g^3\sqrt{g}\log^{5/2}n}\right)$. Moreover, the algorithm does not require a drawing of $G$ as part of the input.
\end{lemma}

\begin{proof}
This is immediate by Lemmas \ref{lem:grid_minor} and \ref{lem:thomassen}.
\end{proof}

Note that computing a large grid minor in the graph is not enough. We need to make sure that a large number of terminals can reach the interface of the grid minor. The following Lemma will provide for us the desired grid minor.
The proof of this Lemma is deferred to Appendix \ref{sec:missing-proofs}.

\begin{lemma}\label{lem:minor_free_flat}
Let ${\cal F}$ be some minor-closed family of graphs, let $\alpha\leq 1$, and $\beta>0$.
Suppose that there exists a polynomial-time algorithm which given, some $G'\in {\cal F}$ and some $\alpha$-node-well-lined set $X'$ in $G'$, outputs some $r'\times r'$ flat grid minor $\Gamma'$ in $G'$, for some $r'=\Omega(\alpha |X'|/\beta)$.
Then there exists a polynomial-time algorithm which, given some $G\in {\cal F}$ and some $\alpha$-node-well-lined set $X$ in $G$, outputs some $r \times r$ flat grid minor $\Gamma$ in $G$, for some integer $r = \Omega(\alpha|X| / \beta)$, and a family of $\lambda r$ node-disjoint paths in $G$ connecting $X$ to the interface of $\Gamma$, for some constant $0<\lambda<1$.
\end{lemma}

\begin{lemma}\label{lem:grid_connected}
Let $G$ be an undirected graph of genus $g$, and let $\alpha \leq 1$. Let $X$ be an $\alpha$-node-well-linked set in $G$. One can, in polynomial time, find some $r \times r$ flat grid minor $\Gamma$ in $G$, for some integer $r = \Omega\left(\frac{\alpha|X|}{g^3\sqrt{g}\log^{5/2}n}\right)$, and a family of $\lambda r$ node-disjoint paths connecting $X$ and the interface of $\Gamma$, for some $0 < \lambda \leq 1$.
\end{lemma}

\begin{proof}
This is immediate by combining Lemmas \ref{lem:flat_grid_minor} and \ref{lem:minor_free_flat}.
\end{proof}

Now by Lemmas \ref{lem:reduction}, \ref{lem:eulerian}, and \ref{lem:grid_connected} we get the following result.

\begin{theorem}\label{thm:genus}
Let $G$ be a graph of genus $g$. There is a polynomial time randomized  algorithm that, with high probability, achieves an $\Omega\left(\frac{1}{g^3\sqrt{g}\log^{5/2}(n)}\right)$-approximation with congestion $5$ for \sdndp instances in $G$, where $h$ is an integer dependent only on $H$.
\end{theorem}

\section{Minor Free Graphs}\label{sec:minor-free}

In this section we present the flat grid minor construction for minor-free graphs.
We first consider the problem on nearly embeddable graphs, and when we extend our solution to arbitrary minor-free graphs by dealing with sums of constant size.

\subsection{Nearly Embeddable Graphs}\label{sec:nearly}
In this subsection we work on nearly embeddable graphs.
 First we reduce the problem to the case of zero apices.

\begin{lemma}[Reduction to $(0,g,k,p)$-nearly embeddable graphs]\label{lem:reduction1}
Suppose that there is a polynomial time algorithm for \sdndp in $(0,g,k,p)$-nearly embeddable graphs that achieves a $\beta$-approximation with congestion $c$. Then there is a polynomial time algorithm for \sdndp in $(a,g,k,p)$-nearly embeddable graphs that achieves a $\beta/a$-approximation with congestion $c$.
\end{lemma}

\begin{proof}
Let $G$ be an $(a,g,k,p)$-nearly embeddable graph, and suppose that we are given a \sdndp instance $M = \{s_1t_1,\cdots,s_mt_m\}$ in $G$. Let $A \subseteq V(G)$ be the set of apices in $G$. Let $G' = G \setminus A$. Clearly, $G'$ is a $(0,g,k,p)$-nearly embeddable graph. Let $M' \subseteq M$ be the subset of source-terminal pairs that do not intersect $A$. $M'$ forms a \sdndp instance in $G'$, and thus we can get a $\beta$-approximation solution $S'$ with congestion $c$. Since $|M| \leq |M'| + a$, we have that $S'$ is a $\beta / a$-approximation solution with congestion $c$ for $M$ in $G$, as desired.
\end{proof}

Next we provide an algorithm for \sdndp in $(0,g,k,p)$-nearly embeddable graphs. Let $G$ be an $(0,g,k,p)$-nearly embeddable graph, and let $S$ be the bounded genus subgraph of $G$ on the surface; that is $S$ is obtained from $G$ by deleting all vortices. Let $X \subseteq V(G)$ be the set of terminals. Note that using Lemma \ref{lem:reduction} we can reduce the problem to the case where $X$ is $\alpha$-well-linked for some $\alpha \leq 1$. The following is implicit in \cite{demaine2008linearity}.

\begin{lemma}[Demaine and Hajiaghayi \cite{demaine2008linearity}]\label{lem:tw_S}
Let $t \geq 1$ be the treewidth of $G^\UN$, and let $t'$ be the treewidth of $S^\UN$. Then we have $t' \geq \frac{t}{(p+k)^3}$.
\end{lemma}


\begin{lemma}\label{lem:minor_free_flat_grid_minor}
One can in polynomial time find a $r \cross r$ flat grid minor $\Gamma$ in $G^\UN$,with $r = \Omega\left(\frac{t}{g^3 \sqrt{g} (p+k)^3 \log^{5/2}n}\right)$. 
\end{lemma}

\begin{proof}
By Lemma \ref{lem:tw_S} we have that the treewidth of $S^\UN$ is at least $\frac{t}{(p+k)^3}$. $S^\UN$ is a graph of Euler genus $g$, and thus by Lemma \ref{lem:flat_grid_minor} we get the desired result.
\end{proof}

\begin{lemma}\label{lem:one_vortex_main}
One can in polynomial time find some $r \cross r$ flat grid minor $\Gamma$ in $G^\UN$, for some integer $r = \Omega\left(\frac{t}{g^3 \sqrt{g} (p+k)^3 \log^{5/2}n}\right)$, and a family of $r$ node-disjoint paths connecting $X$ and the interface of $\Gamma$.
\end{lemma}

\begin{proof}
This is immediate by Lemmas \ref{lem:minor_free_flat_grid_minor} and \ref{lem:minor_free_flat}.
\end{proof}

Now by combining Lemmas \ref{lem:reduction}, \ref{thm:main_crossbar}, \ref{lem:one_vortex_main}, the crossbar construction and routing scheme in Section \ref{sec:crossbar}, we get the following result.


\begin{lemma}\label{lem:(0,g,1,p)}
Let $G$ be a $(0,g,k,p)$-nearly embeddable graph. There is a polynomial time randomized  algorithm that, with high probability, achieves an $\Omega\left(\frac{1}{g^3 \sqrt{g} (p+k)^3 \log^{5/2}n}\right)$-approximation with congestion $5$ for \sdndp instances in $G$, where $h$ is an integer dependent only on $H$.
\end{lemma}

\begin{theorem}\label{thm:embeddable}
Let $G$ be a $(a,g,k,p)$-nearly embeddable graph.
There is a polynomial time randomized  algorithm that, with high probability, achieves an $\Omega\left(\frac{1}{ag^3 \sqrt{g} (p+k)^3 \log^{5/2}n}\right)$-approximation with congestion $5$ for \sdndp instances in $G$, where $h$ is an integer dependent only on $H$.
\end{theorem}

\begin{proof}
This follows immediately by Lemmas \ref{lem:(0,g,1,p)} and \ref{lem:reduction1}. 
\end{proof}

\subsection{Dealing with $h$-sums}\label{sec:minorfree}
In this subsection we are going to prove Lemma \ref{lem:flat_in_minor_free}. Let $G$ be a minor-free graph, with treewidth $t$. Let $X \subseteq V(G)$ be the set of terminals. The following is implicit in \cite{demaine2004approximation}.

\begin{lemma}[\cite{demaine2004approximation}]\label{lem:summand}
Let $G_1,G_2$ be two undirected graphs, and let $G_3$ be an $h$-sum of $G_1$ and $G_2$ for some integer $h > 0$. Let $t_1$, $t_2$, and $t_3$ be the treewidth of $G_1$, $G_2$, and $G_3$ respectively. Then we have $t_3 \leq \max\{t_1,t_2\}$.
\end{lemma}

We are now ready to prove our result for computing flat grid minors in minor-free graphs.

\begin{proof}[Proof of Lemma \ref{lem:flat_in_minor_free}]
By using Theorem \ref{thm:robertson}, we get a decomposition of $G^\UN$ into $h$-sums of $h$-nearly-embeddable graphs. By Lemma \ref{lem:summand}, we have that at least one summand $G'$ has treewidth at least $t$. Now $G'$ is a $h$-nearly-embeddable graph with treewidth $t$, and thus by Lemma \ref{lem:one_vortex_main} we get the desired flat grid minor.
\end{proof}

\bibliography{bibfile}
\pagebreak
\appendix
\section{Missing Proofs}\label{sec:missing-proofs}

\begin{proof}[Proof of Lemma \ref{lem:minor_free_flat}]
Let $t$ be the treewidth of $G$. Since $X$ is $\alpha$-node-well-linked in $G$, we have that $t = \Omega(\alpha |X|)$. Let $\Gamma_0$ be an $r' \times r'$ flat grid minor in $G$, for some $r' = \Omega(\alpha |X|/\beta)$. If there is a family of $\lambda r_0$ node-disjoint paths connecting $X$ and the the interface of $\Gamma_0$, then we are done. Otherwise, we will find an \emph{irrelevant} vertex; that is a vertex $v \in V(G)$ such that deleting $v$ from $G$ does not affect the well-linkedness of $X$. Therefore, we can delete $v$ from $G$, and recursively call the process for finding flat grid minors, until we get the desired one.

Suppose that there is not a family of $\lambda r_0$ node-disjoint paths connecting $X$ and the interface of $\Gamma_0$. 
First we find a $r_0' \times r_0'$ sub-grid $\Gamma_0'$ of $\Gamma_0$ such that $r_0' = O(r_0)$ and $\Gamma_0'$ contains at most $\frac{\lambda r_0}{\alpha}$ terminals. For any minor $H$ of $G$, and for every $v \in V(H)$, let $\eta(v) \subseteq V(G)$ be the subset of vertices in $G$ corresponding to $v$. Let also $X_H = X \cap \eta(H)$. 
Since there is not a family of $\lambda r_0$ node-disjoint paths connecting $X$ and the interface of $\Gamma_0$, we can find a cut $C \subseteq E(G)$ in $G$, separating $X_{\Gamma_0}$ and the interface of $\Gamma_0$, with $|C| < \lambda r_0$. Now let $A_1,A_2,\ldots,A_m$ be the connected components of $G \setminus C$ that contain vertices of $X_{\Gamma_0}$ (See Figure \ref{fig:components}). We may assume w.l.o.g.~that $|V(A_1)| \geq |V(A_2)| \geq \ldots \geq |V(A_m)|$. Now let $Y,Z \subset X$ be two disjoint subsets of $X$ of equal size such that $X_{A_1} \subset Y$ and $X_{A_i} \subset Z$ for any $i \in \{2,3,\ldots,m\}$. Since $X$ is $\alpha$-node-well-linked, there exist a family $\cal{P}$ of $|Y|$ paths from $Y$ to $Z$ such that no vertex is in more than $1/\alpha$ of these paths. However, we have $X_{A_1} \subset Y$ and $X_{A_i} \subset Z$ for any $i \in \{2,3,\ldots,m\}$, and thus we have $|V(X_{A_2}) \cup \ldots \cup V(X_{A_m})| \leq |C| \frac{1}{\alpha} < \frac{\lambda r_0}{\alpha}$. Therefore, we can find a $\frac{r_0}{4} \cross \frac{r_0}{4}$ sub-grid $\Gamma_0'$ of $\Gamma_0$ such that $\Gamma_0'$ does not intersect $X_{A_1}$, and moreover there are at most $\frac{\lambda r_0}{\alpha}$ number of terminals in $\eta(\Gamma'_0)$.

\begin{figure}
\begin{center}
\scalebox{0.40}{\includegraphics{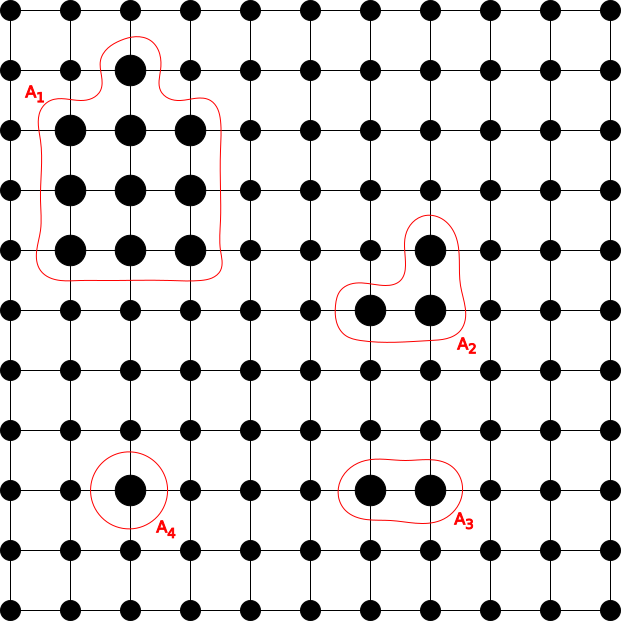}}
\caption{The connected components of $G \setminus C$ in $\Gamma'_0$}\label{fig:components}
\end{center}
\end{figure}

If there is a family of $\lambda r_0'$ node-disjoint paths connecting $X$ and the interface of $\Gamma_0'$, then we are done. Otherwise, we find an irrelevant vertex. We use a similar technique as in \cite{kawarabayashipolylogarithmic}. Let $\Gamma_0''$ be the $r''_0 \times r''_0$ sub-grid of $\Gamma_0'$ obtained by deleting the first and last $r_0'/4$ rows and columns of $\Gamma_0'$. By the construction, we know that $\Gamma_0''$ contains at most $\frac{\lambda r_0}{\alpha}$ terminals. We may assume w.l.o.g.~that $r'_0$ is a power of $2$, and thus $r''_0$ is a power of $2$ as well. We construct a hierarchical partitioning of $\Gamma''_0$ into smaller sub-grids as follows. For every $i,j \in \{1,2,\ldots,r''_0\}$, let $v_{i,j}$ be the vertex in the $i$'th row and $j$'th column of $\Gamma''_0$. For any $i,j,h \in \{1,2,\ldots,r''_0\}$, let 
\[
H_{i,j,h} = \bigcup\limits_{a = \max\{1,i-h-1\}}^{\min\{i+h,r''_0\}} \bigcup\limits_{b = \max\{1,j-h-1\}}^{\min\{j+h,r''_0\}}\{v_{a,b}\}.
\]
We also define $\ell(H_{i,j,h})$ = 2h. For every $q \in \{0,1,\ldots,\log r''_0\}$, we define two partitions of $\Gamma''_0$ into $q \times q$ sub-grids as follows. Let
\[
{\cal{H}}_{q,1} = \bigcup\limits_{i=0}^{r''_0/2^{q+1}} \bigcup\limits_{j=0}^{r''_0/2^{q+1}} \{H(i2^{q+1},j2^{q+1},2^q)\},
\]
and
\[
{\cal{H}}_{q,2} = \bigcup\limits_{i=0}^{r''_0/2^{q+1}} \bigcup\limits_{j=0}^{r''_0/2^{q+1}} \{H(i2^{q+1} + 2^q,j2^{q+1}+2^q,2^q)\}.
\]

Let ${\cal H} = \bigcup\limits_{q=0}^{\log r''_0} \bigcup\limits_{i=1}^{2} {\cal H}_{q,i}$. For every $H \in {\cal H}$, let $w(H)$ be the number of terminals in $\eta(H)$. Let also $w(\Gamma''_0)$ be the number of terminals in $\eta(\Gamma''_0)$. We say that some $H \in {\cal H}$ is \emph{dense} if $w(H) \geq \ell(H)/100$. Let $\delta(\Gamma''_0)$ be the interface of $\Gamma''_0$. We say that some $v \in V(\Gamma''_0)$ is \emph{good} if $v$ is not contained in any dense $H \in {\cal H}$, and there is no terminals in $\eta(v)$. First we show that there exists a good vertex in $\Gamma''_0$. We count the number of vertices in $\Gamma''_0$ that are contained in at least one dense $H \in {\cal H}$. Let ${\cal H}_{q,j} \in {\cal H}$ for some $q \in \{0,\ldots,\log r''_0\}$ and $j \in \{1,2\}$, and let $H \in {\cal H}_{q,j}$. $H$ is dense if and only if $w(H) \geq \ell(H)/100 = 2^{q+1}/100$. We know that $w(\Gamma''_0) \leq r''_0/10000$, and thus if $2^{q+1} > r''_0/100$, then there are no dense $H \in {\cal H}_{q,j}$. Now suppose that $2^{q+1} \leq r''_0/100$, and thus $q < \log r''_0 - 7$. Let $i \in \{8,\ldots,\log r''_0\}$, and let $q = \log r''_0 - i$. Let $H' \in {\cal H}_{q,1}$. We have that $\ell(H') = 2^{q+1} = r''_0/2^{i-1}$. In order for $H'$ to be dense it must be that $w(H') \geq \frac{\ell(H)}{100} = \frac{r''_0}{100\cdot 2^{i-1}}$. Note that we have $w(\Gamma''_0) \leq r''_0/10000$, and therefore there can be at most $2^{i-1}/100$ dense $H' \in {\cal H}_{q,1}$. With a similar argument, we can show that there can be at most $2^{i-1}/100$ dense $H' \in {\cal H}_{q,2}$. Now we have
\begin{align*}
\left|\bigcup\limits_{H \in {\cal H}: H \textit{ is dense}}H\right| &\leq 2 \cdot \sum\limits_{i=8}^{\log r''_0} \left(\frac{r''_0}{2^{i-1}}\right)^2 \cdot \frac{2^{i-1}}{100} \\
&= \frac{(r''_0)^2}{50} \cdot \sum\limits_{i=8}^{\log r''_0} \frac{1}{2^{i-1}} \\
&< \frac{(r''_0)^2}{50}.
\end{align*}
This means that there exist at least $\frac{49(r''_0)^2}{50}$ vertices in $\Gamma''_0$ that are not contained in any dense $H \in {\cal H}$, and since there are at most $r''_0/10000$ terminals in $\eta(\Gamma''_0)$, there must exist a good vertex in $\Gamma''_0$, as desired. Furthermore, this vertex can be found in polynomial time. Let $v \in V(\Gamma''_0)$ be a good vertex.


We claim that vertices in $\eta(v)$ are irrelevant. For every $q \in \{0,1,\ldots,\log r''_0\}$ and $i \in \{1,2\}$, let $H_{q,i} \in {\cal H}_{q,i}$ be a sub-grid that contains $v$. By the construction, for every $q \in \{0,1,\ldots,\log r''_0\}$, we have that either $d_{\Gamma''_0}(v,\delta(H_{q,1})) \geq 2^{q-1}$ or $d_{\Gamma''_0}(v,\delta(H_{q,2})) \geq 2^{q-1}$. Let $B_q \in \{H_{q,1},H_{q,2}\}$ be such that $d_{\Gamma''_0}(v,\delta(B_q)) \geq 2^{q-1}$. For every $q \in \{1,\ldots,\log r''_0\}$, let $C_q = B_q \setminus B_{q-1}$, and let also $C_{\log r''_0 +1} = V(\Gamma'_0) \setminus V(\Gamma''_0)$ (See Figure \ref{fig:anulus}).

\begin{figure}
\begin{center}
\scalebox{0.40}{\includegraphics{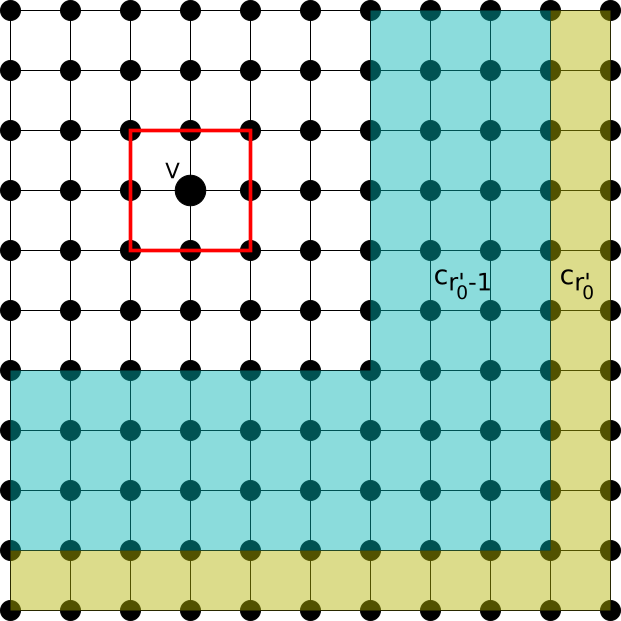}}
\caption{Sets $C_q$}\label{fig:anulus}
\end{center}
\end{figure}

Let $Y,Z \subset X$ be two disjoint subsets of $X$ of equal size. Since $X$ is $\alpha$-node well-linked, we know that there exist a family $\cal{P}$ of $|Y|$ paths from $Y$ to $Z$ such that no vertex is in more than $1/\alpha$ of these paths. If none of these paths use $v$, then we are done. Otherwise, we try to re-route these paths to obtain a new family $\cal{P}'$ of paths, such that no path is using $v$, and no vertex is in more than $1/\alpha$ of the paths in $\cal{P}'$. 
First we look at the paths $P \in {\cal P}$ with both endpoints outside of $\Gamma'_0$; that is the endpoints of $P$ do not belong to $\eta(\Gamma'_0)$. Let ${\cal P}^* \subseteq {\cal P}$ be the set of all such paths. We re-route them in a way such that they do not intersect $\eta(\Gamma''_0)$. Note that by the construction, at most $\lambda r'_0$ of paths in ${\cal P}^*$ can intersect $\eta(\Gamma'_0)$. For these paths, we can re-route their intersection with $\eta(\Gamma'_0)$ in $\eta(\Gamma'_0) \setminus \eta(\Gamma''_0)$, and thus they will not intersect $\eta(\Gamma''_0)$. Now let ${\cal P}^{**} \subseteq {\cal P}$ be the set of paths with one endpoint outside of $\eta(\Gamma'_0)$, and one endpoint inside of $\eta(\Gamma'_0)$. Let $P = (a_1,a_2,\ldots,a_p) \in {\cal P}^{**}$, where $a_1 \notin \eta(\Gamma'_0)$ and $a_p \in \eta(\Gamma'_0)$. Let $a_f \in V(P)$ be the first intersection of $P$ and $\eta(\Gamma'_0)$; that is $f \in \{1,2,\ldots, p\}$ is the minimum number such that $a_f \in \eta(\Gamma'_0)$. Let $P' = (a_f,\ldots,a_p)$. We replace $P$ with $P'$ in ${\cal P}$. Note that again there are at most $\lambda r'_0$ such paths in ${\cal P}$. Now we are only dealing with paths with both endpoints in $\eta(\Gamma'_0)$. For all such paths, we use an inductive argument to re-route them. For any $i,j \in \{1,2,\ldots,\log{r''_0}+1\}$, let ${\cal P}_{i,j} \subseteq {\cal P}$ be the paths with one endpoint in $\eta(C_i)$, and the other endpoint in $\eta(C_j)$. By the construction, for any $i \in \{1,2,\ldots,\log{r''_0}+1\}$, we know that there are at most $2^i / 20$ terminals in $\eta(C_i)$, and thus $|{\cal P}_{i,i}| \leq 2^i / 20$. For all such paths, we can re-route them such that they stay inside $C_i$. We start with ${\cal P}_{\log{r''_0}+1,\log{r''_0}+1}$, and re-route all these paths such that they only use vertices in $C_{\log{r''_0}+1}$. Again, by the construction, we have that $\left|\bigcup\limits_{j = 1}^{\log r''_0} {\cal P}_{\log r''_0 + 1 , j}\right| \leq r''_0/10$. For all $P \in \bigcup\limits_{j = 1}^{\log r''_0} {\cal P}_{\log r''_0 + 1 , j}$, similar to the paths in ${\cal P}^{**}$, we can replace them with paths with one endpoint on the boundary of $C_{\log r''_0}$, and recursively follow the same argument for paths with both endpoints in $\eta\left(\bigcup\limits_{j = 1}^{\log r''_0} C_j\right)$ and so on. Therefore, by applying the same re-routing pattern, we can get a new set of paths ${\cal P}'$ such that no path uses vertex $v$, as desired.

Now let $G_1 = G \setminus v$. Since $v$ is an irrelevant vertex in $G$, we have that $X$ is $\alpha$-node-well-linked in $G_1$, and thus we have that the treewidth of $G_1$ is $\Omega(\alpha |X|)$. Therefore, we can find a $r'_1 \cross r'_1$ flat grid minor $\Gamma_1$ in $G_1$, for some $r'_1 = \Omega\left(\frac{\alpha |X|}{\beta}\right)$. If there exists a family of $\lambda r'_1$ node-disjoint paths connecting $X$ and the interface of $\Gamma_1$, we are done. Otherwise, we recursively follow the same approach to find an irrelevant vertex $v_1$ in $G_1$, and let $G_2 = G_1 \setminus v_1$ and so on. This recursive call stops in $O(n)$ steps, because for each $i \geq 1$, $G_i$ is a graph of treewidth $\alpha |X|$. Therefore, for some $j \geq 1$, we can find a $r_j \times r_j$ flat grid minor $\Gamma_j$ of $G_j$, for some $r_j = \Omega\left(\frac{\alpha |X|}{\beta}\right)$, such that there exists a family of $\lambda r_j$ node-disjoint paths connecting $X$ and the interface of $\Gamma_j$. Note that $\Gamma_j$ is also a flat grid minor of $G$, and this completes the proof.
\end{proof}


\end{document}